\documentclass[11pt]{article}

\usepackage{times}
\usepackage[english]{babel}
\usepackage[latin1]{inputenc}
\usepackage{lmodern}
\usepackage[T1]{fontenc}
\usepackage{amsmath,amssymb}
\usepackage{amsthm}
\usepackage{amsfonts}
\usepackage{hyperref}

\usepackage{tikz} 

\DeclareMathOperator{\acc}{Acc}

\DeclareMathOperator{\Tr}{Tr}
\DeclareMathOperator{\Nil}{Nil}
\DeclareMathOperator{\nil}{nil}

\DeclareMathOperator{\LC}{ ColList}

\DeclareMathOperator{\nt}{Nil}

\DeclareMathOperator{\clx}{collist}

\newcommand{\vgt}[1]{``#1''}
\newcommand{\NN}{\mathbb{N}}

\newcommand{\ap}[1]{\langle #1 \rangle}
\newcommand{\bp}[1]{\left\lbrace #1 \right\rbrace}

\newcommand{\succc}{ \succ_{\mbox{\tiny col}} }
\newcommand{\succT}{\succ_1}
\newcommand{\conc} { {*} }
\newcommand{\setL} { {\mathcal L} }
\newcommand{\ktr}{k\mbox{-}\!\Tr}
\newcommand{\stati}{States}
\newcommand{\hTk}[2]{h_k(#1,#2)}
\newcommand{\hTone}[2]{h_1(#1,#2)}
\newcommand{\htree}[1]{h_k(#1)}

\theoremstyle{plain}
	\newtheorem{theorem}{Theorem}[section]
	\newtheorem{proposition}[theorem]{Proposition}
	\newtheorem{lemma}[theorem]{Lemma}

\theoremstyle{definition}
	\newtheorem{definition}[theorem]{Definition}

\theoremstyle{remark}

\title{Proving termination with transition invariants of height $\omega$}
\author{Stefano Berardi, Paulo Oliva, Silvia Steila}

\begin{document}
\maketitle

\section{Introduction}
The Termination Theorem by Podelski and Rybalchenko \cite{Podelski} states that the reduction relations which are terminating from any initial state are exactly the reduction relations whose transitive closure, restricted to the accessible states, is included in some finite union of well-founded relations. 
An alternative statement of the theorem is that terminating reduction relations are precisely those having a \vgt{disjunctively well-founded transition invariant}.
From this result the same authors and Byron Cook designed an algorithm checking a sufficient condition for termination for a {\sf while-if} program. The algorithm looks for a disjunctively well-founded transition invariant, made of well-founded relations of height $\omega$, and if it finds it, it deduces the termination for the {\sf while-if} program using the Termination Theorem.

This raises an interesting question: What is the status of reduction relations having a disjunctively well-founded transition invariant where each relation has height $\omega$? An answer to this question can lead to a characterization of the set of while-if programs which the termination algorithm can prove to be terminating. The goal of this work is to prove that they are exactly the set of reduction relations having height $\leq \omega^n$ for some $n<\omega$. Besides, if all the relations in the transition invariant are primitive recursive and the reduction relation is the graph of the restriction to some primitive recursive set of a primitive recursive map, then a final state is computable by some primitive recursive map in the initial state.

As a corollary we derive that the set of functions, having at least one implementation in Podelski-Rybalchenko {\sf while-if} language with a disjunctively well-founded transition invariant where each relation has height $\omega$, is exactly the set of primitive recursive functions.

We conjecture that the same result holds for the Terminator Algorithm based on the Termination Theorem: A function has at least one implementation in Podelski-Rybalchenko language which the Terminator Algorithm may catch terminating if and only if the function is primitive recursive. One of the authors is working on a proof of it.

The Termination Theorem is proved in classical logic using Ramsey's Theorem. In order to intuitionistically prove the Termination Theorem we introduced a kind of contrapositive of Ramsey Theorem, the $H$-closure Theorem \cite{Hclosure}. We say that a sequence $s$ is $R$-homogeneous if $s \in H(R)$, where $H(R)$ is defined as follows. 

\begin{center}
{\em Let $R$ be a binary relation on $I$. $H(R)$ is the set of the $R$-decreasing transitive finite sequences on $I$:}
 		\[
 			 \ap{x_1,\dots, x_n} \in H(R) \iff \forall i, j \in [1,n]. i<j \implies x_j R x_i.
 	 	\]
\end{center}

$R$ is $H$-well-founded if $H(R)$ is well-founded by one-step extension. The $H$-closure Theorem says that if $R_1, \dots, R_k$ are $H$-well-founded then $(R_1 \cup \dots \cup R_k)$ is also $H$-well-founded.

$H$-closure, as we said, intuitionistically derives the Termination Theorem. In order to characterize the Termination Theorem in the case of height $\omega$ relations, we first strengthen H-closure as follows. If each $R_i$ has ordinal height less or equal than $\alpha_i$, then the $(R_1 \cup \dots \cup R_k)$-homogeneous sequences have ordinal height less or equal than $2^{\alpha_1 \oplus \dots  \oplus \alpha_k}$, where $\oplus$ is the natural sum of ordinals, defined as the smallest binary function w.r.t. the pointwise ordering which is increasing in both arguments \cite{Carruth}. The proof uses a simulation of the ordering of $H(R_1 \cup \dots \cup R_k)$ in the inclusion ordering over the set of $k$-branching trees, whose branches are decreasing sequences in $R_1 \oplus \dots \oplus R_k$.

Our second step in the characterization of the Termination Theorem is the following. We prove that given a transition relation which is the graph of a partial recursive map restricted to a primitive recursive domain, and given a disjunctively well-founded transition invariant whose relations are primitive recursive and have height $\omega$, we may compute the number of its step and the final state by primitive recursive functions. 

The two proofs are developed in several steps. The first one is to evaluate the ordinal height w.r.t. the reverse-inclusion ordering for the $k$-branching trees of decreasing sequences which belong to an ordinal $\alpha$. This height is $2^\alpha$ if $\alpha$ is a limit ordinal, but for $\alpha$ successor ordinal the expression is more complex. Then, by considering $k$ relations of height $\omega$, we will assign a decreasing labelling for the $k$-ary  trees we use to prove the termination theorem. After that we define an embedding from $H(R_1\cup \dots \cup R_k)$ to the $k$-ary trees. The last step is finding a primitive recursive bound for the Termination Theorem in the case of height $\omega$ relations.

\section{The ordinal height of $k$-ary trees}
We recall some well-known facts about the natural sum. $\alpha \oplus \beta$ is defined as 
\[
	\bigvee\bp{\alpha' \oplus \beta +1, \alpha \oplus \beta'+1 : \alpha' < \alpha, \beta' < \beta}.
\]
By Cantor Normal Form Theorem, each pair of ordinals $\alpha$, $\beta$ may be written as 
\[
\begin{split}
	\alpha&= \omega^{\gamma_1}\cdot n_1 + \dots + \omega^{\gamma_p}\cdot n_p\\
	\beta&= \omega^{\gamma_1}\cdot m_1 + \dots + \omega^{\gamma_p}\cdot m_p
\end{split}
\]
for some $\gamma_1 > \gamma_2> \dots > \gamma_p$ and some $n_1,\ \dots,\ n_p,\ m_1,\ \dots,\ m_p <\omega$. By principal induction over $\alpha$ and secondary induction over $\beta$ we may prove that 
\[
\alpha \oplus \beta=\omega^{\gamma_1}\cdot (n_1+m_1) \oplus \dots \oplus \omega^{\gamma_p}\cdot (n_p+m_p).
\]
As a corollary we deduce that natural sum is commutative and associative. 

\begin{definition}($k$-ary trees on $\alpha$)
\begin{itemize}
\item Let $\alpha$ be any ordinal, we define $\ktr(\alpha)$ as the set of the finite $k$-branching trees labelled with decreasing labels in $[0,\alpha)$. We use  $\nt$ to denote the $k$-branching empty tree.
\item Given $T, \ U \in \ktr(\alpha)$, we define $U \succT T$ as: $U$ is obtained from $T$ by adding one node.
\end{itemize}
\end{definition}

We define a map $\hTk{\cdot}{\alpha}$ computing the ordinal height of the tree $T$ in $\ktr(\alpha)$ w.r.t. $\succT$. $\nt$ has the highest ordinal height w.r.t. $\succT$. Hence $\hTk{\Nil}{\alpha}$ computes the ordinal height of the entire set of such trees. For the results of this paper we only need to know the values of $\hTk{\cdot}{\alpha}$ for $\alpha < \omega^2$. For sake of completeness, however, we will include a study of $h_k(\cdot, \alpha)$ for all $\alpha$. 


\subsection{From $k$-branching trees to the ordinals}
 For any $T$ in $\ktr(\alpha)$ we define a map $\hTk{T}{\alpha}$, then we prove that it computes the ordinal height of $T$ in $\ktr(\alpha)$.
 
 \begin{definition}
 Let $\alpha$ be an ordinal and let $T \in \ktr(\alpha)$. Then we define
\[
	\hTk{T}{\alpha} = \bigoplus\limits_{\nt \mbox{ \scriptsize nodes}}\{\hTk{\nt}{\beta}  : \beta \mbox{ is the label of the father of the $\nt$ node}\}
\]
and 
\[
	\hTk{\nt}{\alpha} = \bigvee\bp{ \hTk{T}{\alpha}+ 1 : T \in \ktr((\alpha)), T \succT \nt}.
\]
where $\oplus$ is the natural sum of ordinals (also known as Hessenberg sum \cite{Carruth}).
\end{definition}

We have to prove that $\hTk{T}{\alpha}$ computes the ordinal height of $T$ in $\ktr(\alpha)$. First we observe that there is an equivalent but simpler description of $\hTk{\nt}{\cdot}$.

Given an ordinal $\alpha$ and a natural number $k$ we define the natural product \cite{Carruth}
\[
	\alpha*k = \alpha \oplus \dots \oplus \alpha
\] 
where there are $k$-many $\alpha$. With $\alpha\cdot k$, instead, we denote the standard product of ordinals.

\begin{lemma}\label{lemma: equiv}
$\hTk{\nt}{\cdot}$ is such that for all $\alpha$
\[
	\hTk{\nt}{\alpha} =  \bigvee\bp{\hTk{\nt}{\beta}*k +1 : \beta < \alpha}.
\]
\end{lemma}
\begin{proof}
Fix an ordinal $\alpha$. We need to prove that
\[
\begin{split}
\bigvee\bp{ \hTk{T}{\alpha}+ 1 : T \in \ktr((\alpha)), T \succT \nt} \\ =\bigvee\bp{ \hTk{\nt}{\beta}*k +1 : \beta < \alpha}.
\end{split}
\]
Let $T \in \ktr((\alpha))$ such that $T \succT \nt$, then $T$ is a root-tree.  Let $\beta$ be the label of the root of $T$. Then, by definition
\[
	\hTk{T}{\alpha} = \hTk{\nt}{\beta}*k.
\]
Hence
\[
\begin{split}
 \bigvee\bp{ \hTk{T}{\alpha}+ 1 : T \in \ktr((\alpha)), T \succT \nt} \\ 
 \leq \bigvee\bp{ \hTk{\nt}{\beta}*k +1 : \beta < \alpha}.
\end{split}
\]

Vice versa, given $\beta \in \alpha$ let $T\succT \nt$ be the root-tree where the root's label is $\beta$. Then
\[
	\hTk{T}{\alpha}= \hTk{\nt}{\beta}*k,
\]
therefore 
\[
	\hTk{\nt}{\beta}*k+1 = \hTk{T}{\alpha}+1.
\]
\end{proof}

Now we prove our thesis about $\hTk{\cdot}{\cdot}$.

\begin{proposition}\label{Proposition: height}
Let $\alpha$ be an ordinal.
\begin{itemize}
\item If $T',\ T \in \ktr((\alpha))$ and $T' \succT T$ then $\hTk{T'}{\alpha} < \hTk{T}{\alpha}$.
\item Let $T \in \ktr((\alpha))$, then $\hTk{T}{\alpha}$ is the ordinal height of $T$ in $\ktr(\alpha)$:
\[
	\hTk{T}{\alpha} = \bigvee \bp{\hTk{T'}{\alpha}+1 : T' \succT T}.
\]
\end{itemize}
\end{proposition}
\begin{proof}
\begin{itemize}
	\item Given $T' \succT T$, let $\gamma$ be the label of the father of the new node $T'$ and $\beta$ be the label of the new node of $T'$. Hence $\hTk{T'}{\alpha}$ has as addends $k$-many $\hTk{\nt}{\beta} $ instead of one $\hTk{\nt}{\gamma}$. Since the labelling is decreasing we have that $\beta < \gamma$. By definition
	\[
		\hTk{\nt}{\gamma}  =  \bigvee\bp{\hTk{\nt}{\beta}*k +1: \beta < \gamma}
	\]
	then $\hTk{\nt}{\gamma} > \hTk{\nt}{\beta}*k$. Since the $\oplus$ is increasing in each argument,  $\hTk{T'}{\alpha} < \hTk{T}{\alpha}$ holds.
	
	\item If $T=\nt$ then the thesis follows by the definition of $\hTk{\nt}{\alpha}$, then let $T \neq \nt$. Thanks to the previous point we have
	\[
		\hTk{T}{\alpha} \geq \bigvee \bp{\hTk{T'}{\alpha}+1 : T' \succT T}.
	\]
	We will prove the other inequality by induction over $\alpha$.
	\begin{itemize}
	\item Assume $\alpha=0$. Then $\hTk{T}{0}= \bigvee\emptyset = 0$.
	\item Assume $\alpha= \beta+1$. If the root of $T$ has label less than $\beta$ then by inductive hypothesis we are done, since $\hTk{T}{\beta+1}= \hTk{T}{\beta}$ and for each $T'\succT T$ $\hTk{T'}{\beta+1}= \hTk{T'}{\beta}$.
	So assume that the root of $T$ has label $\beta$. Let $T_1, \dots, T_k$ be the immediate subtrees of $T$. By definition
	\[
		\hTk{T}{\beta+1} = \bigoplus_{i=1}^{k}\hTk{T_i}{\beta}.
	\]
	We want to prove that for any $\gamma < \hTk{T}{\beta+1}$ there exists $T'\succT T$ such that $\gamma < \hTk{T'}{\beta+1}+1$. 
	Since 
	\[
		\gamma < \hTk{T}{\beta+1} =  \bigoplus_{i=1}^{k}\hTk{T_i}{\beta},
	\]
	by definition of natural sum there exist $\gamma_1,\ \dots,\ \gamma_k$ such that 
	\begin{itemize}
	\item there exists one $j \in [1,k]$ such that $\gamma_j < \hTk{T_j}{\beta}$;
	\item for any $i\in [1,k]$, if $i \neq j$ then $\gamma_i = \hTk{T_i}{\beta}$;
	\item $\gamma < \gamma_1 \oplus \dots \oplus \gamma_k$;
	\end{itemize}
	By induction hypothesis we have that
	\[
		\hTk{T_j}{\beta}= \bigvee\bp{\hTk{U}{\beta}:U\succT T_j}.
	\] 
	Then, since $\gamma_j <\hTk{T_j}{\beta}$, there exists $U\succT T_j$ such that $\hTk{U}{\beta} > \gamma_j$.  Let $T_j'=U$. For any $i\in [1,k]$ if $i \neq j$ we define $T_i'=T_i$.
	Let $T'$ be the tree whose root has label $\beta$ and immediate subtrees $T_1', \dots, T_k'$. By construction $T' \succT T$. Moreover
	\[
		\gamma < \bigoplus_{i=1}^{k}\hTk{T_i'}{\beta} = \hTk{T'}{\beta+1}.
	\]
	Then
	\[
		\hTk{T}{\beta+1} \leq \bigvee \bp{\hTk{T'}{\beta+1} +1 : T' \succT T}.
	\]
	\item  Assume $\alpha$ is limit. Then the root of $T$ has label $\gamma < \alpha$. Hence $\gamma+1<\alpha$, and by inductive hypothesis on $\gamma+1$ we are done, since $\hTk{T}{\alpha}= \hTk{T}{\gamma+1}$ and for each $T'\succT T$ we have $\hTk{T'}{\alpha}= \hTk{T'}{\gamma+1}$.
\end{itemize}
\end{itemize}
\end{proof}

Thanks to Lemma \ref{lemma: equiv} we may define $\hTk{\nt}{\cdot}$ as follows:

\[\hTk{\nt}{\alpha} =
    \begin{cases}
	 0 &\mbox{ if } \alpha=0;\\
	 \hTk{\nt}{\beta}*k + 1 &\mbox{ if } \alpha= \beta+1;\\
	 \bigvee_{\mu < \alpha} \hTk{\nt}{\mu} &\mbox{ if $\alpha$ is limit.}
    \end{cases}
\]

If we may compute $\hTk{\nt}{\alpha}$ then we may compute $\hTk{T}{\alpha}$, that is, by \ref{Proposition: height}, the ordinal height of any $T \in \ktr(\alpha)$.  We may easily compute $\hTk{\nt}{\alpha}$ if either $k=1$ or $\alpha < \omega^{\omega}$. 

\begin{lemma}
	\begin{itemize}
	    \item If $k=1$, $\hTone{\nt}{\alpha} = \alpha$;
		\item \[\hTk{\nt}{m} = \sum_{i=0}^{m-1} k^{i} = \frac{k^m-1}{k-1};\]
		\item \[\hTk{\nt}{\omega\cdot n + m } = \omega^n\cdot k^m + \sum_{i=0}^{m-1} k^{i}\]
	\end{itemize}
\end{lemma}

\begin{proof}
Immediate by induction over $\alpha$.
\end{proof}

%


Now we want to derive what is the value of $\hTk{\nt}{\alpha}$ for any $\alpha$. This analysis is only added for completeness and is not used to derive the results of this paper.

\begin{lemma}
Let $k \geq 2$ and let $\alpha = \lambda + n$ where $\lambda$ is either $0$ or a limit ordinal and $n$ is a natural number, then 
\begin{itemize}
\item if $\lambda = 0$: 
    	\[
		 \hTk{\nt}{n} = \frac{k^{n}-1}{k-1}
	    \]
\item otherwise 
      \[
		\hTk{\nt}{\alpha} =  k^{\alpha}+ \frac{k^{n}-1}{k-1}
	  \]
\end{itemize}
\end{lemma}
\begin{proof}
By induction on $\alpha$. Observe that:
\[
	\hTk{\nt}{\alpha}= \bigvee\bp{\hTk{\nt}{\beta}*k+1:\beta < \alpha}.
\]
Then we have three cases
\begin{itemize}
	\item $\alpha=0$. Then $\hTk{\nt}{0}=0$.
	\item $\alpha=\beta+1$. Then $\alpha= \lambda + (n+1)$, $\beta= \lambda + n$. Hence:
	\[
		\hTk{\nt}{\alpha}= \bigvee\bp{\hTk{\nt}{\gamma}*k+1:\gamma < \alpha}= \hTk{\nt}{\beta}*k+1.
	\]
	\begin{itemize}
		\item If $\alpha$ is finite, also $\beta$ is; then
		\[
			\begin{split}
			\hTk{\nt}{\alpha}&= \hTk{\nt}{\beta}*k+1 = \frac{k^{n}-1}{k-1}*k + 1 \\
			&=\frac{k^{n+1} - k +k-1}{k-1}= \frac{k^{n+1}-1}{k-1}.
			\end{split}
		\]
		\item If $\alpha$ is infinite, also $\beta$ is; then  
		\[
			\begin{split}
			\hTk{\nt}{\alpha}&= \hTk{\nt}{\beta}*k+1 = (k^{\beta}+\frac{k^{n}-1}{k-1})*k + 1\\ &=k^{\beta+1} + \frac{k^{n+1}-1}{k-1}= k^{\alpha} + \frac{k^{n+1}-1}{k-1}.
			\end{split}
		\]
	\end{itemize}
	\item $\alpha$ limit. Then
	\[
	  \begin{split}
	     \hTk{\nt}{\alpha}=& \bigvee\bp{\hTk{\nt}{\beta}*k+1: \beta < \omega} \\ 
		&\vee \ \bigvee\bp{\hTk{\nt}{\beta}*k+1: \omega \leq \beta < \alpha}.
	  \end{split}
	\]
	We may study two different cases:
	\begin{itemize}
		\item If $\alpha=\omega$. Then only the first set is not empty. Moreover it is cofinal in $\omega$. Then $\hTk{\nt}{\omega}= \omega = k^\omega$
		\item If $\alpha > \omega$. Then the first set is cofinal in $\omega$, while the second set is cofinal in $[\omega,  \alpha)$. Then, since $\alpha$ is limit:
		\[
			\hTk{\nt}{\alpha}=\omega \vee \bigvee\bp{k^{\beta}+\frac{k^{n+1}-1}{k-1} +1: \omega \leq \beta < \alpha}= k^{\alpha}.
		\]
	\end{itemize}
\end{itemize}
\end{proof}

Since if $\alpha$ is a limit ordinal, then $2^\alpha= k^\alpha$ for any $k \geq 2$, it follows that if $\alpha$ il limit then 

\[
	\hTk{\nt}{\alpha} =  k^\alpha = 2^\alpha.
\]

Moreover if $\alpha= \omega\cdot k$, we have that $2^\alpha=\omega^k$.

\subsection{Erd\H{o}s trees}

In this subsection we recall the definitions of Erd\H{o}s trees. Erd\H{o}s trees are inspired by the trees used first by Erd\H{o}s then by Jockusch in their proofs of Ramsey \cite{Jockusch}, hence the name. Given $k$ many relations $R_1,\dots, R_k$ we may think of each branch an Erd\H{o}s-tree on $R_1,\dots, R_k$ as a simultaneous construction of all $R_i$-decreasing transitive lists for all $i\in [1,k]$.

In order to formally define an Erd\H{o}s tree we need the definition of colored list: a list of $n$ elements, with an assignment of colors to the $n-1$ edges. Both the empty list and all one-element lists have the empty set of edges, therefore have the empty assignment of colors to edges. 

\begin{definition}
	A colored list $(L, f)$ is a pair, where $L= \ap{x_1, \dots, x_n}$ is a list on $I$ and $f= \ap{c_1, \dots, c_{n-1}}$ is a list on $C$. $\nil = (\ap{}, \ap{})$ is the empty colored list, $\clx(x)= (\ap{x}, \ap{})$ is the colored list whose only node is $x$, and $\LC(C)$ is the set of the colored lists with colors in $C$.
\end{definition}

We can define the relation one-step extension on colored lists: $\succ_c$ is the one-step extension of color $c$ and $\succc$ is the one-step extension of any color. Assume $x \in I$ and $\lambda, \mu \in \LC(C)$:

\begin{itemize} 
\item $\lambda \conc_c \clx(x)  \succ_c \lambda$;
\item $\lambda \succc \mu$ if $\lambda \succ_c \mu$ for some $c \in C$,
\end{itemize}

where, $c \in C$ and $\conc_c$ is the composition of color $c$ of two colored lists by connecting the last element of the first list with the first of the second list with an edge of color $c$. Formally:
\[
\begin{split}
	&\nil \conc_c \lambda = \lambda \conc_c \nil = \lambda;\\
 	&\mbox{if } L, M \neq \nil,  \ (L, f) \conc_c (M, g) = (L \conc M, f \conc \langle c \rangle \conc g).
\end{split}
\] 
An $R_1, R_2, \dots, R_k$-colored list is an attempt to build simultaneously one $R_h$-decreasing list for each $h \in [1,k]$.

\begin{definition}
	Let $C=[1,k]$. $(L,f)\in \LC(C)$ is a $R_1,R_2, \dots, R_k$-colored list if $L=\ap{x_1, \dots, x_n}$, $f=\ap{c_1,\dots, c_{n-1}}$, and 
	\[
		\forall i \in [1,n-1].(c_i=h \implies (\forall j \in [1,n]. i<j\implies (x_j R_h x_i))).
	\]
	$\LC(R_1, R_2, \dots, R_k) \subseteq \LC(C)$ is the set of $(R_1,R_2, \dots, R_k)$-colored lists.
\end{definition}

Take any $R_1, R_2, \dots, R_k$-colored list $L$. Let $L'$ be the sublist of $L$ consisting of all elements which either are followed by a branch of color $h$, or are at the end of $L$. Then $L'$ is an $R_h$-decreasing list, because by definition each  element of $L'$ is connected by $R_h$ to each element of $L$ after it, hence, and with more reason, is connected by $R_h$ to each element of $L'$ after it. Thus, as we anticipated, any $R_1, R_2, \dots, R_k$-colored list L defines simultaneously one $R_h$-decreasing list for each $h \in [1,k]$.

\begin{definition}
	A $k$-ary tree $T$ is a set of colored lists on $I$, such that:
	\begin{enumerate}
		\item $\nil$ is in $T$;
		\item If $\lambda \in T$ and $\lambda \succc \mu$, then $\mu \in T$;
		\item Each list in $T$ has at most one one-step extension for each color $c \in C$: if $\lambda_1, \lambda_2, \lambda \in T$ and $\lambda_1, \lambda_2 \succ_c \lambda$, then $\lambda_1 = \lambda_2$.
	\end{enumerate} 
For all sets $\setL \subseteq \LC(C)$ of colored list, $\ktr(\setL)$ is the set of $k$-ary trees whose branches are all in $\setL$.
\end{definition}

We need also the one-step extension $\succT$ between $k$-ary tree; $T' \succT T$ if $T'$ has one leaf more than $T$. 

\begin{definition}[One-step extension for $k$-ary trees]
If $T$ is a $k$-ary tree and $\lambda \in T$ and $\mu \succ_c \lambda$ for some $c \in [1,k]$ and $\lambda' \succ_c \lambda$ for no $\lambda' \in T$, then 
$$T \cup \bp{\mu} \succT T$$
\end{definition}

We call an {\em Erd\H{o}s-tree over $R_1,\dots, R_k$} any $k$-ary tree whose branches are all in $\LC(R_1,\dots, R_k)$. 

In \cite{Hclosure} we proved that each one-step step extension in a $H(R_1\cup \dots \cup R_k)$ may be simulated as a one-step extension of some Erd\H{o}s tree on $(R_1,\dots, R_k)$, that is, as adding a child to some $R_1, R_2, \dots, R_k$-colored list of the tree. From the well-foundation of the set $\ktr(\LC(R_1,\dots, R_k))$ of Erd\H{o}s trees we derived $H$-closure Theorem. Now we want to use the ordinal bound for $k$-ary trees to derive an ordinal bound for $H(R_1\cup \dots \cup R_k)$, in the case $R_1, \dots,  R_k$ all have height $\omega$.

\subsection{Labelling an Erd\H{o}s tree}

So let now consider only the Erd\H{o}s trees (in $\ktr(\LC(R_1, \dots, R_k))$), following the notation of \cite{Hclosure}. From now on we will assume that $R_1, \dots, R_k$ have height $\omega$. 

We may associate to each node the $k$-uple $(y_1, \dots, y_k)$ of the integer heights of the node w.r.t. the relations $R_1, \dots, R_k$. Thus, it is enough to compute an upper bound to the ordinal height of an Er\H{o}s tree $T$ w.r.t $\succT$ in the following case:  the set of nodes of $T$ is $I=\omega \times \omega \dots \times \omega$ ($k$-many times), and for all $h \in [1,k]$ $(y_1, \dots, y_k) R_h (y_1', \dots, y_k')$ is equivalent to $y_h <y_h'$. If we are able to give an upper bound in this case, we are able to give an upper bound whenever $R_1, \dots, R_k$ have height $\omega$. There is no obvious guess about such an height: if $(y_1, \dots, y_k)$ is a node and $(y_1', \dots, y_k')$ is the child number $h$ of the node, all we do know is that $y_h > y_h'$. The remaining components of $(y_1, \dots, y_k)$ and $(y_1', \dots, y_k')$ may be in any relation. In fact, it is not even evident that all branches are finite: this result requires, and is immediately equivalent to, the Ramsey Theorem.

Our first task will be to label the nodes of any $R_1, \dots, R_k$-list in a decreasing way, by ordinals $< (\omega \cdot k)$. To this aim, we first introduce the notion of $i$-node.

\begin{definition}
Let $T$ be in $\ktr(\LC(R_1, \dots, R_k))$. Assume $y = (y_1, \dots, y_k)$ is a node of $T$. Let $i \in \NN$. 
\begin{enumerate}
\item $(y_1, \dots, y_k)$ is an $i$-node of $T$ w.r.t. $h_1, \dots, h_i$ if the branch from the root to $y$ has exactly $i$-many different colors $h_1, \dots, h_i$.
\item Assume $y$ is an $i$-node w.r.t. $h_1, \dots, h_i$. For any $j \in [1,i]$, we denote by $ y^{h_j}= (p^{h_j}_1, \dots, p^{h_j}_k)$ the lowest proper ancestor of $y=(y_1, \dots, y_k)$  in the branch from the root to the node, which is followed by an edge of color $h_j$.
\end{enumerate}
\end{definition}

Every node is an $i$-node for some $i \in [0,k]$, and $i>0$ if and only if the node is not the root. By definition, if a node of the branch $z$ is followed by an edge of color $h$ then all descendants of $z$ in the branch are smaller w.r.t. $R_h$. Thus, if $y=(y_1, \dots, y_k)$ is an $i$-node of $T$ w.r.t. $h_1, \dots, h_i$, then:
\begin{itemize}
\item for any proper ancestor $z$ of $y$ we have $y R_{h_j} z$, for some $j \in [0,i]$;
\item for any $j \in [0,i]$, there exists an ancestor $z$ of $y$ such that $y R_{h_j} z$.
\end{itemize}

The color $h_j$ denotes the edge from the child number $h_j$. Thus, for any $j \in [1,i]$, the node $y$ is a descendant of the child number $h_j$ of the node $y^{h_j}$.

Then we may label the node $(z_1, \dots, z_k)$  in a decreasing way with ordinals $< \omega \cdot k$, as follows.

\begin{definition} (the labelling $\alpha$).
Let $ T \in \ktr(\LC(R_1, \dots, R_k))$ and let  $(z_1, \dots, z_k)$  be a node of  $T$:
\begin{itemize}
\item if $(z_1, \dots, z_k)$ is the root of the tree, then 
\[
	\alpha((z_1, \dots, z_k)) = \max_{i\in[1,k]}\bp{z_i+1} \oplus \omega*(k-1);
\]
\item if, for some $j>0$,  $(z_1, \dots, z_k)$ is a $j$-node w.r.t. $h_1, \dots h_j$
\[
	\alpha((z_1, \dots, z_k)) = p^{h_1}_{h_1} \oplus \dots \oplus p^{h_j}_{h_j} \oplus \omega*(k-j).
\]
\end{itemize}
\end{definition}

We may observe that each node has label less than the one of its father.

\begin{lemma}
The labelling $\alpha$ is decreasing w.r.t. the father/child relation.
\end{lemma}
\begin{proof}
Let $(z_1, \dots, z_k)$ be a node of the tree and assume that $(y_1, \dots, y_k)$ is its father, then we have three possibilities.
\begin{itemize}
\item If the father is the root then there exists $j \in k$ such that:
\[
\begin{split}
\alpha((z_1, \dots, z_k)) &=  y_{h_j} \oplus \omega*(k-1) \\
 &<\max_{i\in[1,k]}\bp{y_i+1} \oplus \omega*(k-1) = \alpha((y_1, \dots, y_k)) 
\end{split}
\]
\item If, for some $j>0$, the father is a $j$-node and the child is still a $j$-node, then the child is connected to its father with the relation $R_{h_i}$ for some $i \in [1,j]$. Hence the lowest $h_i$-ancestor of the child is its father (whose $h_i$ component is less than the one of its $h_i$-ancestor), then the label decreases. 
\item If, for some $j>0$, the father is a $j$-node and the child is a $j+1$-node then the labels decreases since we have an \vgt{infinite component that becames finite}.
\end{itemize}
\end{proof}

We will prove that the $R_1 \cup \dots \cup R_k$-homogeneous sequences are interpretable in Erd\H{o}s trees where the branching are decreasing with respect to $R_1 \oplus \dots \oplus R_k$.

\begin{lemma}\label{lemma: aus1}
	If $T'$, $T$ are Erd\H{o}s trees and $T' \succT T$ then $\hTk{T'}{\omega\cdot k} < \hTk{T}{\omega\cdot k}$.
\end{lemma}
\begin{proof}
	It follows by Proposition \ref{Proposition: height}, since $\alpha(\cdot)$ is a decreasing labelling. 
\end{proof}

Moreover, if each relation has height $ \omega $ then we have 
\[
	 \hTk{\cdot}{\omega\cdot k}: \ktr(\LC(R_1, \dots, R_k)) \rightarrow \omega^{k}+1;
\] 
where  $\hTk{\nt}{\omega\cdot k}=\omega^k$ and for each $T \in \ktr(\LC(R_1, \dots, R_k)) \setminus\{ \nt \}$ $\hTk{T}{\omega\cdot k}<\omega^k$. 

Then we have a primitive recursive function from the set of Erd\H{o}s trees over $R_1, \dots, R_k$ in $\omega^k+1$ such that if $T' \succT T$ then  $\hTk{T'}{\omega\cdot k} < \hTk{T}{\omega\cdot k}$. From this fact, and the fact that we may embed any transitive subset of $R_1 \cup \dots \cup R_k$ in the set of Erd\H{o}s trees over $R_1, \dots, R_k$,  we will derive our results about the Termination Theorem.

From now on, we will use $\htree{T}$ instead of $\hTk{T}{\omega\cdot k}$.
\section{A primitive recursive bound for a special case of the Termination Theorem}

In order to state our result about primitive recursive sets we need the following definition.

\begin{definition}
Let $D$ be any subset of $\stati$ and $R$ any binary relation on $\stati$. $R$ is the graph of a primitive recursive function restricted to a primitive recursive domain $D$ if
\begin{enumerate}
\item $D$ is primitive recursive and
\item $R$ is the graph of a primitive recursive function $f: \stati \rightarrow \stati$ restricted to $D$: i.e.
\[
	R=\bp{(x,f(x)) : x \in D}.
\]
\end{enumerate}
\end{definition}

We may formally  state our main result as follows: given a reduction relation which is a the graph of a primitive recursive function restricted to a primitive recursive domain such that there exists a disjunctively well-founded transition invariant whose relations are primitive recursive and have height $\omega$, there exists a primitive recursive bound to the number of reductions steps.

\subsection{Finding a primitive recursive bound with the lexicographic order}

\begin{lemma}\label{lemma: banale}
	If $\sigma: \NN \rightarrow \NN$ is primitive recursive and there exist $m,n \in \NN$ such that $m<n$ and $\sigma(m) < \sigma(n)$ then
	\[
		\exists p \in [m,n-1] (\sigma(p) < \sigma(p+1)).
	\]
\end{lemma}

\begin{proof}
	Since the statement is decidable we may reason by contradiction and de Morgan. Assume the opposite: if $\sigma$
	\[
		\forall p \in [m,n-1] (\sigma(p) \geq \sigma(p+1)),
	\]
	Then $\sigma(m) \geq \sigma(n)$. Contradiction.
\end{proof}

We denote with $\preccurlyeq_k$ the lexicographic order of $\NN^k$.

Given a function $g$, define $g^n(x) = g \circ g^{n-1} (x)$. We may observe that if $g$ is primitive recursive, also $H(n,x)= g^{n+1}(x)$ is. In fact:

\[
	H(n,x)=
	\begin{cases}
	x & n=0\\
	g(H(n-1, x)) &\mbox{otherwise}.
	\end{cases}
\]

\begin{lemma}\label{ausiliario}
	For each $\sigma: \NN \rightarrow \NN^k$ primitive recursive, there exists $g: \NN \rightarrow \NN$ primitive recursive such that
	\[
		\forall n \exists m \in [n, g(n)](\sigma(m) \preccurlyeq_k \sigma(m+1)).
	\]
\end{lemma}
\begin{proof}
	By induction on $k$.
	If $k=1$ we put
	\[
		g(n):= n +\sigma(n)+1.
	\]
	Let $n \in \NN$, we want to prove that there exists $m \in [n,n+\sigma(n)]$ such that $\sigma(m) \leq \sigma(m+1)$. Suppose, by contradiction, 
	\[
		\forall m \in[n, n+\sigma(n)+1] \sigma(m) > \sigma(m+1).
	\]
	then we obtain a sequence of $\sigma(n)+2$ many  decreasing natural numbers from $\sigma(n)$. Contradiction.

	Assume that it holds for $k$. We will prove it for $k+1$.
	Let
	\[
		\begin{split}
			\sigma: \NN &\rightarrow \NN^{k+1}\\
			n &\mapsto (\sigma_1(n), \sigma_k(n))
		\end{split}
	\]
	primitive recursive, where $\sigma_1:\NN \rightarrow \NN$ and $\sigma_k: \NN \rightarrow \NN^k$. Then also $\sigma_k$ is primitive recursive then by inductive hypothesis there exists $g_k$ such that
	\[
		\forall n \exists m \in [n,g_k(n)] (\sigma_k(m) \preccurlyeq_k \sigma_k(m+1)). 
	\]
	Put $H(0,x)=x$ and, for any $n>0$:
	\[
		H(n,x) = g_k(H(n-1,x)+1)
	\]

	\[
		g:=  H(\sigma_1(n) + 2, n).
	\]
	We want to prove that 
	\[
		\forall n \exists m \in [n,g_k^{\sigma_1(n)+2}(n)](\sigma(m) \preccurlyeq_{k+1} \sigma(m+1)).
	\]

	If $\exists i < j \in [n, g_{k}^{\sigma_1(n)+2}(n)]$ such that $\sigma_1(i) < \sigma_1(j)$ then by Lemma \ref{lemma: banale} we obtain $\exists p \in [i,j-1]$  $\sigma_1(p) < \sigma_1(p+1)$. It follows that $\sigma(p) \preccurlyeq \sigma(p+1)$ in the lexicographic order and we are done.  Otherwise assume that
	\[
		\forall i, j \in [n, g_k^{\sigma_1(n)+2} (n)] (\sigma_1 (i) \geq \sigma_1(j)).
	\]
	We apply the inductive hypothesis over $\sigma_k$ and the disjoint intervals: 
		\[
	    ([n, g_k(n)], [g_k(n)+1, g_k(g_k(n)+1)], \dots).
     	\] 

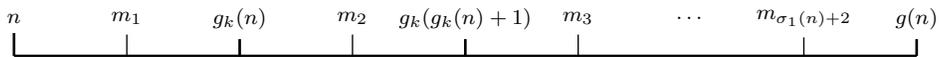
	\[
	\begin{tikzpicture}
	\node (n)        at (-7.5,0) {\scriptsize $n$};
	\node (m1)       at (-6,0)   {\scriptsize$m_1$};
	\node (gk)       at (-4.5,0) {\scriptsize$g_k(n)$};
	\node (m2)       at (-3,0)   {\scriptsize$m_2$};
	\node (gkgk)     at (-1.5,0) {\scriptsize$g_k(g_k(n)+1)$};
	\node (m3)       at (0,0)    {\scriptsize$m_3$};
	\node (dots)     at (1.5,0)  {\scriptsize$\dots$};
	\node (ms)       at (3,0)    {\scriptsize$m_{\sigma_1(n)+2}$};
	\node (gks)      at (4.5,0)  {\scriptsize$g(n)$};

	\draw [line width=1](-7.5,-0.5) -- (4.5,-0.5);
	\draw [line width=1](n) -- (-7.5,-0.5);
	\draw [line width=0.5](m1) -- (-6,-0.5);
	\draw [line width=1](gk) -- (-4.5,-0.5);
	\draw [line width=0.5](m2) -- (-3,-0.5);
	\draw [line width=1](gkgk) -- (-1.5,-0.5);
	\draw [line width=0.5](m3) -- (0,-0.5);
	\draw [line width=0.5](ms) -- (3,-0.5);
	\draw [line width=1](gks) -- (4.5,-0.5);
	\end{tikzpicture}
	\]

	We obtain that there are some $m_1< m_2 < \dots < m_{\sigma_1(n)+2}$ such that $\sigma_k(m_i) \preccurlyeq_k \sigma_k(m_i+1)$. Moreover, by assumption on $[n, g_k^{\sigma_1(n)+2}(n)]$ we have $\sigma_1(n) \geq \sigma_1(m_1) \geq \sigma_1(m_2) \geq \dots \geq \sigma_1(m_{\sigma_1(n)+2})$. Then there exists $i\in [1,\sigma_1(n)+2]$ such that $\sigma_1(m_i) = \sigma_1(m_{i+1})$.
	Hence 
	\[
		\sigma_1(m_i) = \sigma_1(m_i+1)= \sigma_1(m_{i+1}).
	\]
	Since, by inductive hypothesis $\sigma_k(m_i) \preccurlyeq_k \sigma_k(m_i+1)$ it follows 
	\[
		\sigma(m_i) \preccurlyeq_{k+1} \sigma(m_i +1).
	\]
\end{proof}

\subsection{From $H(R_1\cup \dots\cup R_k)$ to Erd\H{o}s trees}

Now we want to define a primitive recursive function from $H(R_1 \cup \dots \cup R_k )$ to the Erd\H{o}s trees, in order to find a primitive recursive bound for the number of step of a sequence in $H(R_1 \cup \dots \cup R_k )$.

Let $\clx(x)= \ap{\ap{x},\ap{}}$ be the colored list including only $x$, and $*_i$ the junction of two colored lists with an edge of color $i$. Assume that for any node $x$ of $T$ we have $y _i x$ for some $i$. Let $c(y,x)=i$ if $i$ is the first integer in $[1,k]$ such that $yR_ix$. We denote with $T_i$ the $i$-th immediate subtree of $T$. Then we may recursively define:
\[
	h(T,y)= 
	\begin{cases}
		\clx(y) & T=\nt\\
		\clx(r)*_i h(T_i, y)  &  \mbox{ if } $r$ \mbox{ is the root of }T \mbox{,} \ i=c(r,y).
	\end{cases}
\]

Let define $E$ primitive recursive from $H(R_1\cup \dots\cup R_k)$ to the set of the  Erd\H{o}s trees as follows. 
\[
	E(\ap{x_1,\dots, x_n}) = 
	\begin{cases}
	\nil & n=0\\
	E(\ap{x_1,\dots, x_{n-1}}) \cup \bp{h(E(\ap{x_1,\dots, x_{n-1}}), x_n)}  & n>0;
	\end{cases}
\]

\begin{lemma}\label{lemma: aus2}
	If $L' \succ L$ then $E(L') \succT E(L)$.
\end{lemma}
\begin{proof}
If $L' \succ L$ then $L' = L * \ap{y}$ for some $y$. If follows that $E(L')= h(E(L),y)$, i.e. we add a leaf to a leaf of $E(L)$ (following the idea of Erd\H{o}s). Then $E(L') \succT E(L)$.
\end{proof}

\subsection{Main Theorem}

We define a primitive recursive increasing map $f^*: H(R_1\cup\dots\cup R_k) \rightarrow \omega^k+1$, from $R_1\cup\dots\cup R_k$-homogeneous sequences to ordinals, by $f^*(s) = \htree{E(s)}$.

\begin{lemma}\label{lemma: fstar}($f^*$ is increasing)
	If $L' \succ L$ then $f^*(L') < f^*(L)$.
\end{lemma}
\begin{proof}
	By applying Lemma \ref{lemma: aus2} $E(L') \succT E(L)$. By Lemma \ref{lemma: aus1} 
	\[
		f^*(L') = \htree{E(L')} < \htree{E(L)}= f^*(L).
	\]
\end{proof}

Let $P$ be a program. We define a reduction relation $R$ as in Podelski and Rybalchenko paper \cite{Podelski}. Let $t$  be a computation which behaves like $R$ until it reaches a final state and then it repeats this state, i.e. if $x$ is a final state $t(x)=x$.

\begin{theorem}\label{theorem:main}
	Assume that $P$ is such that $ R^+ \cap (\acc \times \acc) = R_1 \cup \dots \cup R_k$, where
	\begin{enumerate}
	\item the complement of the set of the final states of $R$ ($\stati \setminus F$) is a primitive recursive set;
	\item $R$ is the graph of a primitive recursive function $f: \stati \setminus F \rightarrow \stati$: i.e.
	\[
		R=\bp{(x,f(x)) : x \in \stati\setminus F}.
	\]
	\item $R_1, \dots, R_k$ are primitive recursive relations and have height $\omega$.
	\end{enumerate}
	Then there exists $g': \stati \rightarrow \NN$ such that $t^{g'(s)}(s)= t^{g'(s)+1}(s)$.
\end{theorem}

\begin{proof}

	Observe that thanks to hypothesis 1 and 2 we have that $t$ is primitive recursive, while thanks to the third one $f^*$ is primitive recursive.
	
	Let $\phi(x):= f^*(\ap{s_0,t^{1}(s_0) ,\dots, t^{x}(s_0)})$. Then $\phi: \NN \rightarrow \NN^k$, since the input list for $f^*$ cannot be the empty list. Moreover $\phi(x)$ is primitive recursive. In fact let
\[
	\theta(x)= \begin{cases} 
	\ap{s_0} &\mbox{if } x=0;\\
	\theta(x-1)*t^x(s_0) &\mbox{ otherwise }.
	\end{cases}
\]
It follows that $\phi$ is a composition of primitive recursive function, so is primitive recursive:

\[
	\phi(x)= f^*(\theta(x)).
\]

Thanks to Lemma \ref{ausiliario} there exists $g$ primitive recursive
\[
	\forall n \exists m \in [n, g(n)] (\phi(m) \preccurlyeq_k \phi(m+1)).	
\]
Put $n=0$. Then there exists $m \in [0, g(0)]$  such that

\[ f^*(\theta(m)) = \phi(m) \preccurlyeq_k \phi(m+1) = f^*(\theta(m+1)).\]

Observe that, since  $R^+ \cap (\acc \times \acc) = R_1 \cup \dots \cup R_k$ and $R^+ \cap (\acc \times \acc)$ is transitive, for any $m \in \NN$ if $t^{m-1}(s_0)$ is not a final state then $\theta(m) \in H(R_1\cup \dots \cup R_k)$, and if $t^{m}(s_0)$ is not a final state then $\theta(m+1) \succc \theta(m)$.

By Lemma \ref{lemma: fstar} we ontain that $t^{m}(s_0)$ is not a final state, then \[f^*(\theta(m)) \succ_k f^*(\theta(m+1)),\] contradicting \[f^*(\theta(m)) \preccurlyeq_k f^*(\theta(m+1)).\] Thus, $t^{m}(s_0)$ is a final state.
\end{proof}

\section{Vice versa} 
In this section we will prove the vice versa of the theorem \ref{theorem:main}: if $f$ is a primitive recursive function then there exists a program $P$ which evaluates $f$ such that $P$ has a disjunctively well-founded transition invariant $T$ composed only by primitive recursive relations of height at most $\omega$. For short we say that a transition invariant is primitive recursive and has height $\omega$ if it is composed only by primitive recursive relations with height $\omega$.  

In order to prove this result we will use the following notation. Given a program $P$ of the form

\begin{verbatim}
 int f(int x1, ..., int xn)
 {  int v1, ..., vm;
    CODE 
    return r;
 }
\end{verbatim}

we define the code $\ap{f(x_1,\dots, x_n),y}$ as follows:

\begin{verbatim}
 int z1=x1;
 ...
 int zn=xn;
 CODE[x1/z1,..., xn/zn];
 y=r;
\end{verbatim}

Now we may prove the main result of this section.
\begin{theorem}
Let $f$ be a primitive recursive function, then there exists a program $P$  which evaluates $f$ such that:
\begin{itemize}
\item its transition relation is a graph of a primitive recursive function restricted to a primitive recursive domain;
\item $P$ has a primitive recursive transition invariant disjunctively well-founded of height $\omega$.
\end{itemize} 
\end{theorem}

\begin{proof}
By induction on primitive recursive functions.
\begin{enumerate}
\item {CONSTANT FUNCTION 0.} If $\forall x.f(x)=0$
\begin{verbatim}
 int f(int x){
  return 0;
 }
\end{verbatim}
Then $R=\emptyset$ and so one transition invariant is $T = \emptyset$.
\item {SUCCESSOR FUNCTION.} If $\forall x.f(x)=x+1$.
\begin{verbatim}
 int f(int x){
  return x+1;
 }
\end{verbatim}
As above $R=\emptyset$ and so one transition invariant is $T = \emptyset$.
\item {PROJECTION FUNCTION.} Let $i \in n$. If 
\[
	\forall x_1, \dots, x_n.f(x_1, \dots, x_n)=x_i.
\]
\begin{verbatim}
 int f(int x1, ... int xn){
  return xi;
 }
\end{verbatim}
As above $R=\emptyset$ and so one transition invariant is $T = \emptyset$.
\item {COMPOSITION.} Let 
\[
	\forall x_1, \dots, x_n. f(x_1, \dots x_n)= h(g_1(x_1, \dots, x_n), \dots g_k(x_1, \dots, x_n),
\]
	
where for $h, \ g_1, \dots,\ g_k$ there exists a program $P_h, \ P_{g_1}, \dots \ P_{g_k}$ with one transition invariant disjunctively well-founded $T_h, \ T_{g_1}, \dots \ T_{g_k}$ of height $\omega$ by induction hypothesis. Then we define $P$ as follows
\begin{verbatim}
 int f(int x1, ... int xn){
  int a, y1, ..., yk, res,... ; 
  a=1;
  < g1(x1, ... , xn), y1>
  ...
  a=a+1;
  < gk(x1, ... , xn), yk>
  a=a+1;
  < h(y1, ... , yk), res>
  return res;
 }
\end{verbatim}
where $a, y_1, \dots, y_n, res$ and the variables in $\ap{g_i(x_1, \dots , x_n), y_i}$ and in $\ap{h(x_1, \dots , x_n), res}$ are fresh. Observe that for plainness it is better to declare all the variables used  in $P$ in the first instruction. In this way each state of $P$ is a sequence of values of all the variables which appear in $P$. 
Let 
\[
	T=\bp{(\ap{a,\bar{s}}, \ap{a',\bar{s}'}) : a < k+1, a < a' < k+2 }
\]
and for any $i \in [1,k]$
\[
	T_{g_i}^*= \bp{(\ap{i,\bar{s}}, \ap{i,\bar{s}'}) : \bar{s} T_{g_i}\bar{s}' }
\]
and 
\[
	T_{h}^*= \bp{(\ap{k+1,\bar{s}}, \ap{k+1,\bar{s}'}) : \bar{s} T_{h}\bar{s}' }.
\]

Hence $T_f= T \cup T_{g_1}^* \cup \dots \cup T_{g_k}^* \cup T_h^*$ is a transition invariant disjunctively well-founded of height $\omega$ for $P$. In fact if 
$(\bar{s},\bar{s}') \in R^+\cap(\acc \times \acc)$, we have one of the following possibilities:
\begin{itemize}
\item They are two states in the same functions $g_i$ for some $i$ or $h$. Then $(\bar{s},\bar{s}') \in T_{g_i}$ or $(\bar{s},\bar{s}') \in T_{h}$ by inductive hypothesis. 
\item They are two states in two different functions then the variable $a$ in the first state has a smaller value then the second one. So $(\bar{s},\bar{s}') \in T$. 
\end{itemize}
This proves that it is a transition invariant. It has height $\omega$ since $T_{g_1}, \dots, T_{g_k}, T_h$ have height $\omega$ and $T$ has height $k+2$.   Moreover it is disjunctively well-founded since $T_{g_1}, \dots, T_{g_k}, T_h$ are and $T$ is well-founded.

Moreover the transition relation of $P$ is a graph of a primitive recursive function restricted to a primitive recursive domain; this is an exercise in complexity theory.

\item {PRIMITIVE RECURSION.} Let 
\[
\begin{cases}
	&f(0, x_1, \dots, x_k)= h(x_1, \dots, x_k),\\
	&f(S(y),x_1, \dots, x_k)=g(y, f(y,x_1, \dots, x_k),x_1, \dots x_k).
\end{cases}
\]
where for $h$ and $g$ there exist two programs $P_h, P_{g}$ with a transition invariant disjunctively well-founded $T_h, \ T_{g}$ of height $\omega$ by induction hypothesis. 
Let $P_g$ be as follows:
\begin{verbatim}
int g(int y, int q, int x1, ... , int xn)
{
 CODE
 return r;
}
\end{verbatim}
 
and let $P_h$ be:

\begin{verbatim}
int h(int x1, ... , int xn)
{
 CODE
 return r;
}
\end{verbatim}

Then we define $P$ as:

\begin{verbatim}
 int f(int y, int x1, ... int xk){
  int z, z1, ..., zk, ...; 
  z=0;
  < h(x1, ... , xk), w>
  z1=x1;
  ...
  zk=xk;
  while (z < y)
  { z=z+1; 
    CODE[ y/z, q/w, x1/z1,..., xk/zk];
    w=r;
  }
  return w;
 }
\end{verbatim}
where $z, z1, \dots, z_k,w$ and the variable of $\ap{h(x_1,\dots , x_k), w}$ and of
	CODE$[y/z, q/w, x_1/z_1,\dots, x_k/z_k]$ are fresh variables. Observe that, again,  for plainness it is better to declare all the variables used in the first instruction. 
Let
\[
	T_{h}^*= \bp{(\ap{0,\bar{s}}, \ap{0,\bar{s}'}) : \bar{s} T_{h}\bar{s}' }.
\]
In this case the transition invariant is
\[
\begin{split} 
	T_h^* \cup &\bp{ ((z, y, \bar{s}),(z,y, \bar{s}')) : (\bar{s}, \bar{s}') \in T_g, z < y, y \in \omega} \cup\\
	&\cup \bp{((z,y, \bar{s}), (z+1,y, \bar{s}') : y \in \omega, z < y )}
\end{split}
\]
Let us call the second one $T_1$ and the third one $T_2$. It is a transition invariant since if $(\bar{s},\bar{s}')\in R^+\cap(\acc \times \acc)$ then we have one of the following possibilities:
\begin{itemize}
\item the two states are states of $h$, then $(\bar{s},\bar{s}')\in T_h$;
\item the two states are states in same round of the new while then  $(\bar{s},\bar{s}')\in T_1$;
\item the two states are states in two different rounds of the new while then $(\bar{s},\bar{s}')\in T_2$.
\end{itemize}
Then it is a transition invariant. It has height $\omega$ since $T_h$ is,  $T_1$ is the union of relations of height $\omega$ since $T_g$ is, and $T_2$ has height $\omega$. Moreover it is disjunctively well-founded since $T_h$ and $T_g$ are and $T_2$ is well-founded.

Moreover, the transition relation of $P$ is a graph of a primitive recursive function restricted to a primitive recursive domain; as above, this is an exercise in complexity theory.
\end{enumerate}
\end{proof}

\bibliographystyle{plain}
\bibliography{bibliografia}

\begin{thebibliography}{1}

\bibitem{Hclosure}
Stefano Berardi and Silvia Steila.
\newblock {Ramsey Theorem as an intuitionistic property of well founded
  relations}.
\newblock submitted, 2014.

\bibitem{Carruth}
Philip Carruth.
\newblock Arithmetic of ordinals with applications to the theory of ordered
  abelian groups.
\newblock 48(4):223--334.

\bibitem{Jockusch}
Carl G.~Jockusch Jr.
\newblock {Ramsey's Theorem and Recursion Theory}.
\newblock {\em J. Symb. Log.}, 37(2):268--280, 1972.

\bibitem{Podelski}
Andreas Podelski and Andrey Rybalchenko.
\newblock Transition invariants.
\newblock In {\em LICS}, pages 32--41, 2004.

\end{thebibliography}
\end{document}